\documentclass[12pt,english]{amsart}
\usepackage[T1]{fontenc}
\usepackage[utf8]{inputenc}
\usepackage{geometry}
\usepackage{pgfplots}
\pgfplotsset{compat=1.12}
\usepackage{array}

\usepgfplotslibrary{fillbetween}
\usetikzlibrary{patterns}
\usepackage{textcomp}
\usepackage{amstext}
\usepackage{amsthm}
\usepackage{amssymb}
\usepackage{graphicx}
\usepackage[unicode=true,pdfusetitle,
 bookmarks=true,bookmarksnumbered=false,bookmarksopen=false,
 breaklinks=true,pdfborder={0 0 1},backref=false,colorlinks=false]
 {hyperref}

\makeatletter
\numberwithin{equation}{section}
\numberwithin{figure}{section}
\theoremstyle{plain}
\newtheorem{thm}{\protect\theoremname}[section]
\theoremstyle{plain}
\newtheorem{prop}[thm]{\protect\propositionname}
\newtheorem{lemma}[thm]{Lemma}
\newtheorem*{remark}{Remark}
\usepackage{xurl}
\usepackage{graphicx}
\usepackage{color}
\usepackage{xcolor}

\makeatother

\providecommand{\propositionname}{Proposition}
\providecommand{\theoremname}{Theorem}
\title{Sturm-Hurwitz theorem for quantum graphs}

\author{Ram Band}
\address{Faculty of Mathematics \\ Technion - Israel Institute of Technology \\ 32000 Haifa \\ Israel}
\address{Institut f\"ur Mathematik \\ Universit\"at Potsdam \\ 14476 Potsdam Germany }
\email{ramband@technion.ac.il}

\author{Philippe Charron}
\address{Section de Math\'{e}matiques \\ Universit\'e de Gen\`eve \\ 1211 Gen\`eve}
\email{philippe.charron@unige.ch}
\begin{document}
\global\long\def\ep{\varepsilon}%

\global\long\def\la{\lambda}%

\global\long\def\rmi{\mathrm{i}}%

\global\long\def\rme{\mathrm{e}}%

\global\long\def\Om{\Omega}%

\global\long\def\N{\mathcal{\mathbb{N}}}%

\global\long\def\graph{G(m,\ep)}%

\global\long\def\id{\mathrm{id}}%
\begin{abstract}
	We prove upper and lower bounds for the number of zeroes of linear combinations of Schr\"odinger eigenfunctions on metric (quantum) graphs. These bounds are distinct from both the interval and manifolds. We complement these bounds by giving non-trivial examples for the lower bound as well as sharp examples for the upper bound. In particular, we show that even tree graphs differ from the interval with respect to the nodal count of linear combinations of eigenfunctions. This stands in distinction to previous results which show that all tree graphs have to same eigenfunction nodal count as the interval.
\end{abstract}

\maketitle

\section{Introduction}

\subsection{Historical background}\hfill
\medskip

The rigorous study of the zero set of eigenfunctions of second-order differential
operators originated in the 19th century, with Sturm's oscillation
theorem on the interval being the first major result. The subject
now encompasses graphs and manifolds, and while certain results hold
true for all these objects, other properties of the nodal set are
highly dependant on the dimension.

\medskip

For instance, Sturm's theorem asserts that the n-th eigenfunction
of any Sturm-Liouville operator on an interval has exactly $n-1$ zeroes \cite{Stu_jmpa36}. 

\medskip

For quantum graphs, the $n$-th eigenfunction of the Laplacian with Neumann-Kirchhhoff continuity conditions has between $n-1$ and $n-1+\beta$ zeroes\footnote{Under the assumption that $\lambda_n$ is simple and the eigenfunction is non-zero at interior vertices.}   \cite{Ber_cmp08,GSW03}, where $\beta$ is the number of cycles of the graphs. For example, we get that the $n$-th eigenfunction of a tree graph (which is a graph with $\beta=0$) has exactly $n-1$ zeroes, similarly to the interval. Furthermore, tree graphs are the only graphs such that the $n$-th eigenfunction has exactly $n-1$ zeroes for all $n$ \cite{Ban_ptrsa14}.

\medskip

On manifolds, Courant's theorem \cite{Courant1923} states that when we remove
the zero set of the $n$-th eigenfunction of the Laplace-Beltrami operator on a connected smooth manifold
$M$, we are left with at most $n$ connected components (also known as nodal domains). This implies that the zero set has at most $n-1$ connected components.  However, there are examples of eigenfunctions on the sphere and the square where the nodal set has exactly one connected component \cite{Stern_Bemerkungen25}.

\medskip

A lesser-known generalization of Sturm's result was published in the
same year \cite{Stu_jmpa36a}: let $F:=\sum_{i=j}^{k}a_{i}f_{i}$, where $f_{i}$ is the $i$-th Sturm-Liouville eigenfunction on an interval $[a,b]$ with Dirichlet boundary conditions. Then,
$F$ has at least $j-1$ and at most $k-1$ zeroes in $(a,b)$. We refer to this as the Sturm-Hurwitz theorem. An interesting survey on the story of this result and its proofs can be found in \cite{BERARD202028,BERARD202027}.

\medskip

In the case of manifolds, no such bounds can be found in full generality: there are 
metrics on the torus and the sphere and linear combinations
of eigenfunctions of the corresponding Laplace-Beltrami operator such that their nodal set has infinitely many connected components \cite{101093imrnrnz181}, \cite{BerHelCha2021}. We also note that the Sturm-Hurwitz theorem is true for the isotropic quantum harmonic operator in dimension two \cite[proposition $B.1$]{BerHelCha2021}.

\medskip

\subsection{Definitions}\hfill
\medskip

We will use the following definitions troughout the paper:

\medskip

\begin{itemize}
	\item A graph $\Gamma = \Gamma(V,E)$ with vertices $V$ and edges $E$. We will assume throughout the paper that the graph $\Gamma$ is connected and has a finite number of edges, each of which has finite length.
	\item Each edge $e \in E$ is identified with the interval $[0,l_{e}]$, where $l_e$ is the length of the edge.
	\item We set $\mathcal{E}_v$ as the multi-set of edges $e$ which are connected to $v$, where each loop appears twice (once per direction).
	
	\item The degree of a vertex $v$ is defined as $|\mathcal{E}_v|$, and will be denoted by $deg(v)$. 
	\item Vertices of degree one are boundary vertices. The set of such vertices will be denoted $V_b$.
	\item Vertices of degree two or higher are inner vertices. The set of such vertices will be denoted $V_i$.
	\item $\beta$ denotes be the first Betti number of the graph, so that $\beta = |E|-|V|+1$ if $\Gamma$ is connected. It also represents the number of independant cycles, or the number of edges that one has to cut to turn the graph into a tree.
	\item $L^2(\Gamma) = \bigoplus\limits_{e \in E} L^2(e)$, $C^1(\Gamma) = \bigoplus\limits_{e \in E} C^1(e)$ and $\mathcal{H}^2(\Gamma) = \bigoplus\limits_{e \in E} \mathcal{H}^2(e)$ 
	\item Let $f \in C^1(\Gamma)$. If $f$ is continuous at inner vertices, for any $v \in V$ we define $f(v)$ as the common limit of $f(x)$ as $x$ approaches $v$ on any edge $e \in \mathcal{E}_v$.
	\item The normal derivative of a function $f$ at a vertex $v$ in the direction of $e \in \mathcal{E}_v$ will be denoted by $\partial_{e}f(v)$ and defined by taking the right limit at $0^{+}$ of $f'$ when the edge $e$ is identified with $[0,l_{e}]$ and $v$ is mapped to zero. 
	\item We set $\mathcal{H}_\Gamma$ as space of functions $f\in \mathcal{H}^{2}(\Gamma)$ with
	the following continuity conditions at the vertices:
	\begin{itemize}
		\item[(1)] If $v \in V_b$, $f(v)=0$ (Dirichlet boundary condition). 
		\item[(2a)] For any $v \in V_i$, $f$ is continuous at $v$.  
		\item[(2b)] If $v \in V_i$, then $\sum\limits _{e \in \mathcal{E}_v}\partial_{e}f(v)=0$. 
		\item (2a) and (2b) together will be called Neumann-Kirchoff continuity conditions. 
\end{itemize}
	\item We say that a function has a degenerate edge if it is identically zero on that edge.
	\item Let $W:\Gamma \to \mathbb{R}$ be a $C^1$ function. We define the Schr\"odinger operator \linebreak $H_W : \mathcal{H}^2(\Gamma) \to L^2(\Gamma)$, $H_W = -\frac{\partial^2}{\partial x^2} + W $.
	\item The operator $H_W$ restricted to $\mathcal{H}_\Gamma$ is self-adjoint and has an increasing sequence of eigenvalues $\lambda_1 \leq \lambda_2 \ldots$, numbered with multiplicity (see for instance \cite{Mugnolo_book}). 
	
	\item To each eigenvalue  $\lambda_i$ we associate an eigenfunction $f_i$ such that any two different $f_i$ are orthogonal in $L^2(\Gamma)$.
		\item The number of zeroes of a function $f$ which are distinct from the boundary vertices will be denoted by $N(f)$.

	\item We say that a graph $\Gamma$ is $W$-generic if all eigenfunctions of $H_W$ do not vanish at any inner vertex. We note that by continuity of eigenfunctions, this assumption implies that any eigenfunction of $H_W$ does not have a degenerate edge and that every eigenvalue of $H_W$ is simple. Indeed, if there is a multiple eigenvalue, given any vertex $v \in V_i$ it is always possible to choose a function in the linear space of eigenfunctions associated to this eigenvalue that is zero at $v$.

\end{itemize}
\subsection{Acknowledgements}\hfill
\medskip

We thank Bernard Helffer for pointing out this problem and encouraging us to explore it for quantum graphs. We also thank Gregory Berkolaiko, Amir Sagiv and Uzy Smilansky for stimulating discussions on the problem.

The two authors were supported by ISF (grant number 844/19). R. B. was also supported by the Binational Foundation Grant (grant no. 2016281). P. C. was also supported by a Zuckerman Fellowship grant and by SNSF.

\section{Results}

\medskip

Our main result is a two-sided bound on $N(F)$.

\begin{thm}
\label{thm1} Let $\Gamma$ be a $W$-generic graph with first Betti number $\beta$. Let $f_{k}$ be the
eigenfunctions of $H_W=-\frac{\partial^2}{\partial x^2}+W$ with Dirichlet boundary conditions and Neumann-Kirchhoff continuity conditions, $k_{i}$ be a strictly increasing
sequence and \linebreak $F(x)=\sum_{i=1}^{M}a_{i}f_{k_{i}}(x)$ where each $a_i$ is not zero.
We have the following bounds:

\begin{equation}
k_1-1-(M-1)\left( |V_b| + 2\beta -2 \right) \leq N(F)\leq k_M-1+\beta+(M-1)\left( |V_b| + 2\beta -2 \right)\,.
\end{equation}

\end{thm}

\medskip

	\begin{remark}
	First, by setting $\Gamma$ as the unit interval, we recover Sturm's original theorem. Secondly, it is likely that the same bounds hold if we put Neumann boundary conditions on $V_b$ on the graph. 
	
	\end{remark}

We claim that there exist tree graphs that saturate the upper bound in Theorem \ref{thm1}.

\begin{thm}\label{thm3}
	For any $s, M >0$, there exists a $0$-generic star graph with $s+1$ edges such that for any $L \leq M$, there exists a linear combination of the first $L$ eigenfunctions of $-\frac{\partial^2}{\partial x^2}$ with exactly $L-1+(L-1)(s-1)$ zeroes.
\end{thm}

\medskip
\begin{remark}{Theorem \ref{thm3} shows that trees are different from the interval in terms of the nodal count of linear combinations of eigenfunctions. In particular, even certain star graphs equipped with the Laplacian and no potential can possess linear combinations whose nodal count is substantially higher than the number of zeroes of the highest eigenfunction. This behavious is completely different than the nodal count of individual eigenfunctions. As a comparison, all tree graphs have exactly the same eigenfunction\footnote{Here, we mean a single eigenfunction nodal count rather than a linear eigenfunction nodal count.} nodal count as the interval \cite{Ban_ptrsa14,PokPryObe_mz96,Sch_wrcm06}. }
\end{remark}

\medskip

 We now give examples of graphs with linear combinations that have much less zeroes than the lowest eigenfunction in the linear combination:

\begin{thm}\label{thm6}
	For any $m \geq 2$, there exists a $0$-generic graphs with $\beta =m$ and  $b(m) \in \mathbb{R}$ such that 
	
	\begin{align}
		& N(f_2) = m \, , \nonumber\\
		& N(f_3) = 2 \, , \nonumber\\
		& N(f_2 + b(m)f_3) = 1 = N(f_2) -  \frac{1}{2}\sum\limits_{v \in V_i}(deg(v)-2)\,.
	\end{align}
\end{thm}

\medskip

While this is not a proof that the lower bound in Theorem \ref{thm1} is sharp, it shows that a linear combination $F$ can have much less zeroes than $f_{k_1}$, unlike for the interval.

\medskip

\section{Proof of Theorem \ref{thm1}}

Let $\Gamma$ be a $W$-generic graph, $H_W = -\frac{\partial^2}{\partial x^2} +W$ and $f_k$ be the eigenfunctions of $S_W$.

\medskip

Let $F(x)=\sum_{i=1}^{M}a_{i}f_{k_i}(x)$. We define $g: \Gamma \times \mathbb{R} \to \mathbb{R}$ as 

\begin{equation}
	g(x,y)=\sum_{i=m}^{M}a_{i}e^{-\la_{k_i}y}f_{k_i}(x)\,.
\end{equation} 

\medskip

Then, $g(x,0)=F(x)$. The function $g$ is a solution to the following equation:

\begin{equation}
\frac{\partial g}{\partial y}=\frac{\partial^{2}g}{\partial x^{2}}-W(x)g\,.
\end{equation}

\medskip

For fixed $y$, we will denote $F_y(x):=g(x,y)$. Note that for $y=0$, $F_0 = F$.

\medskip

Since $\la_{k_i}>\la_{k_{i-1}}$ for all $i$, $\lim\limits_{y \to -\infty}F_y e^{\lambda_{k_M} y} =   a_M f_{k_M} $ and $\lim\limits_{y \to +\infty}F_y e^{\lambda_{k_1} y} =  a_1 f_{k_1} $. Therefore, as $y \to -\infty$ the zeroes of $F_y$ will converge to the zeroes of $f_{k_M}$ and as $y \to +\infty$ the zeroes of $F_y$ will converge to the zeroes of $f_{k_1}$. 

\subsection{Strategy of the proof of Theorem \ref{thm1}}\hfill
\medskip

First, we will show in section \ref{section:finiteness} that for any $y$, $F_y$ has discrete zeroes.

\medskip

Secondly, we will show in section \ref{section:isolated} that $g$ has no isolated zeroes.

\medskip

Then, our strategy will be the following: starting at $y=-\infty$, we will follow $N(F_y)$ as $y$ increases to $0$. To every zero of $f_y$, we can associate at least locally a nodal line of $g$, or more than one if $(x,y)$ is a singular point of $g$. We will then look in sections \ref{section:events} and \ref{section:innervertex} at all possible local behaviours of the nodal lines of $g$ which would cause $N(F_y)$ to change as $y$ increases.

\medskip

Finally, we will combine all these observations in section \ref{sectfinal} to complete the proof.

\subsection{Local behaviour of $F_y$ near a zero}\label{section:finiteness}\hfill
\medskip

We want to show that a nodal line of $g$ cannot have a horizontally flat part. This is implied by the following lemma:

\begin{lemma}\label{lemma:finiteness}
	For any $y \in \mathbb{R}$, $F_y$ cannot be zero on an open set.
\end{lemma}

\medskip
We will show the following statement, which implies Lemma \ref{lemma:finiteness}:

\begin{prop}\label{propfin}
Let $W \in C^1((a,b))$, $\lambda_j$ be a strictly increasing sequence of real numbers and $f_j$ be non-zero $C^2$ solutions to the equation $-f_j''+ W(x)f_j = \lambda_j f_j $ on an interval $(a,b)$. Then, for any integer $M$, integers $k_i$ and real numbers $a_i$, $1 \leq i \leq M$, the function $G(x)= \sum\limits_{i=1}^M a_i f_{k_i}$ cannot be identically zero on an open subset of $(a,b)$.
\end{prop}

\begin{proof}
We will prove this by induction on $M$. If $M=1$, by standard Sturm-Liouville theory (see for instance \cite[Lemma $1.3.1$]{LevSar_book}) $f'' = \lambda f + W f$ cannot be zero on an open set inside an edge without being identically zero on the edge, which contradicts the $W$-genericity of the graph. Now, assume that this is true for $M-1$. Let $\sum\limits_{i=1}^M a_i f_{k_i}\equiv 0$ on an open set $U \subset (a,b)$. Then, for any $x \in U$, we have the following:

\begin{align}
0 &= \sum_{i=1}^M a_i f_{k_i}''(x) \, ,\\
&= \sum_{i=1}^M a_i \la_{k_i}f_{k_i}(x) + \sum_{i=1}^M a_i W(x)f_{k_i}(x) \, ,\\
&= \sum_{i=1}^M a_i \left(\la_{k_i}-\la_{k_1}\right)f_{k_i}(x) + \left(W(x) + \la_{k_1}\right)\sum_{i=1}^M a_i f_{k_i}(x)\, ,\\
&= \sum_{i=2}^M a_i \left(\la_{k_i}-\la_{k_1}\right)f_{k_i}(x)\, .
\end{align}

 Since $x \in U$ was taken arbitrarily, we can conclude that $\sum_{i=2}^M a_i \left(\la_{k_i}-\la_{k_1}\right)f_{k_i}(x)$ is identically zero on $U$, which contradicts our induction hypothesis since it is a linear combination of $M-1$ eigenfunctions. Therefore, a linear combination of $M$ eigenfunctions cannot be zero on an open set in $(a,b)$, which completes the proof of Proposition \ref{propfin}, and hence of Lemma \ref{lemma:finiteness}.
\end{proof}

\subsection{Isolated zeroes}\label{section:isolated}\hfill

We now show that $g$ cannot have isolated zeroes. 

\begin{lemma}\label{lemma:isolated}
The function $g$ does not have an isolated zero inside an edge.
\end{lemma}
\begin{proof}

To prove this, we will use the following special case of the maximum principle for parabolic equations:

\begin{thm}\label{thm:angenent}\cite[Theorem D]{Angenent1988}
	Let $u: [x_0,x_1] \times [0,Y]$ be a solution of \linebreak $\frac{\partial u}{\partial y} = \frac{\partial^2 u}{\partial x^2}+c(x,y)u$ with $c \in L^\infty$ and for a fixed $y$ set $u_y(x) := u(x,y)$. Assume that $u(x_0,y) \neq 0$ and $u(x_1,y) \neq 0$ for all $y \in [0,Y]$. Then, for any fixed $y \in (0,Y)$, the number of zeroes of $u_y(x)$ is finite. Also, if $u(x',y')= \frac{\partial u}{\partial x}(x',y')=0$, then for any $0 <y_1 < y' < y_2 < Y$, $N(u_{y_1}) > N(u_{y_2})$.
	
\end{thm}

Assume that $(x_0,y_0)$ is an isolated zero of $g$. Let $U:=[x_0-\delta,x_0+\delta]\times[y_0-\ep,y_0+\ep]$ be a small enough neighbourhood of $(x_0,y_0)$ such that $(x_0.y_0)$ is the only zero of $g$ in $U$. Let $u$ be the restriction of $g$ to $U$ and $u_y(x) := u(x,y)$. For any $y\in (y_0-\ep,y_0) \cup (y_0, y_0+\ep) $, $N(u_{y})=0$. However, $(x_0,y_0)$ has to be a critical point of $u$ and so $\frac{\partial u}{\partial x}(x_0,y_0)=0$. By theorem \ref{thm:angenent}, $N(u_{y_0 + \ep/2})< N(u_{y_0 - \ep/2})$, a contradiction since both are assumed to be zero. This proves Lemma \ref{lemma:isolated}.
\end{proof}

\subsection{Local behaviour of nodal lines of $g$}\hfill
\medskip

In order to understand the local behaviour of nodal lines of $g$ up to vertices in the interior of each edge, we will use another form of the maximum principle for parabolic equations:

\begin{prop}\label{maximumprinciple}
	Let $W:[a,b]\to \mathbb{R}$ be $C^1$ and $g:[a,b]\times [c,d]\to \mathbb{R}$ satisfy \linebreak $\frac{\partial g}{\partial y}=\frac{\partial^{2}g}{\partial x^{2}}-W(x)g$ on $(a,b)\times (c,d)$ and continuous on $[a,b]\times (c,d)$. Let \linebreak $(x_0,y_0) \in [a,b]\times (c,d)$. If $g(x_0,y_0)=0$, then there is at most one nodal line exiting from $(x_0,y_0)$ as $y$ increases.
\end{prop}

\medskip

The proof is classical, we include it here for the sake of completeness.

\begin{proof}

\medskip

 Since the potential $W$ is $C^1$, let $w = \min\limits_{[a,b]} W(x)$. Let $G(x,y) = g(x,y) e^{-(w+1)y}$. 

\medskip

The function $G$ solves the following equation:

\begin{align}
	\frac{\partial G}{\partial y} = \frac{\partial^2 G}{\partial x^2} - (W-w+1)G\, .
\end{align}

\medskip

Let $x_1(y)$ and $x_2(y)$ be two adjacent nodal lines of $G$ for $y$ in some interval \linebreak $(c',d') \subset (c,d)$. Furthermore, assume that $G>0$ between $x_1(y)$ and $x_2(y)$. Let $y \in (c',d')$ and $x(y)$ be a local maximum of $h_y(x) := G(x,y)$ between $x_1(y)$ and $x_2(y)$. 

\medskip

Since $x(y)$ is a local maxima on the interval $(x_1(y),x_2(y))$, then $\frac{\partial^2 G}{\partial x^2}(x(y),y) \leq 0$. Also, since $W(x) -w \geq 0$ for any $x \in [a,b]$, $\frac{\partial G}{\partial y}(x(y),y) \leq -G(x(y),y) < 0$. Therefore, as $y$ increases, the value at the local maximum of $h_y(x)$ has to strictly decrease. 

\medskip

Hence, as $y$ decreases, the absolute value of any local maximum between two adjacent zeroes of $h_y$ has to increase. Therefore, two adjacent zeroes cannot collide as $y$ decreases since that would imply that the absolute value of any local maximum would converge to $0$.  Since the zeroes of $G$ and $g$ are exactly the same, this proves Proposition \ref{maximumprinciple}. 
\end{proof}

\subsection{Events}\label{section:events}\hfill
\medskip

As $y$ increases, here are all the possibilities:

\vspace{1cm}

\begin{center}
\scalebox{0.8}{
\begin{tabular}{| m{20em} | m{14em} |}
\hline
	(E1) A nodal line curves upwards & 
		\begin{tikzpicture}[scale=0.25]
			\draw[ultra thick] (1,5) .. controls (2.5,2) .. (4,5);
			\draw[]  (5,-1) -- (0,0) -- (0,8) -- (5,7)  ;
		\draw[dotted, very thick]  (5,-1) -- (5.5, -0.5) -- (4.5, 0.5) -- (5.5, 1.5) -- (4.5,2.5) -- (5.5,3.5) -- (4.5,4.5) -- (5.5,5.5) -- (4.5,6.5) -- (5,7);
			\draw[dashed]  (0,0) -- (-3,-1) -- (-3,7) -- (0,8) node[anchor=north west]{$v_i$} ;
			\draw[dashed] (0,0) -- (-1,2) -- (-1,10) -- (0,8);
		\end{tikzpicture}
		\\ \hline
(E2) A nodal line splits into two or more lines & 
	\begin{tikzpicture}[scale=0.25]
		\draw[ultra thick] (1,5) .. controls (2.5,2) .. (4,5);
		\draw[ultra thick] (2.5,1) -- (2.5,2.85);
		\draw[]  (5,-1) -- (0,0) -- (0,8) -- (5,7)  ;
		\draw[dotted, very thick]  (5,-1) -- (5.5, -0.5) -- (4.5, 0.5) -- (5.5, 1.5) -- (4.5,2.5) -- (5.5,3.5) -- (4.5,4.5) -- (5.5,5.5) -- (4.5,6.5) -- (5,7);
		\draw[dashed]  (0,0) -- (-3,-1) -- (-3,7) -- (0,8) node[anchor=north west]{$v_i$} ;
		\draw[dashed] (0,0) -- (-1,2) -- (-1,10) -- (0,8);
	\end{tikzpicture}
	\, \, 
	\begin{tikzpicture}[scale=0.25]
		\draw[ultra thick] (1,5) .. controls (2.5,2) .. (4,5);
		\draw[ultra thick] (2.5,1) -- (2.5,5);
		\draw[]  (5,-1) -- (0,0) -- (0,8) -- (5,7)  ;
		\draw[dotted, very thick]  (5,-1) -- (5.5, -0.5) -- (4.5, 0.5) -- (5.5, 1.5) -- (4.5,2.5) -- (5.5,3.5) -- (4.5,4.5) -- (5.5,5.5) -- (4.5,6.5) -- (5,7);
		\draw[dashed]  (0,0) -- (-3,-1) -- (-3,7) -- (0,8) node[anchor=north west]{$v_i$} ;
		\draw[dashed] (0,0) -- (-1,2) -- (-1,10) -- (0,8);
	\end{tikzpicture}
 \\ \hline
(E3) Two or more nodal lines collide with each other & 
	\begin{tikzpicture}[scale=0.25]
		\draw[ultra thick] (1,1) .. controls (2.5,4) .. (4,1);
		\draw[ultra thick] (2.5,1) -- (2.5,5);
		\draw[]  (5,-1) -- (0,0) -- (0,8) -- (5,7)  ;
		\draw[dotted, very thick]  (5,-1) -- (5.5, -0.5) -- (4.5, 0.5) -- (5.5, 1.5) -- (4.5,2.5) -- (5.5,3.5) -- (4.5,4.5) -- (5.5,5.5) -- (4.5,6.5) -- (5,7);
		\draw[dashed]  (0,0) -- (-3,-1) -- (-3,7) -- (0,8) node[anchor=north west]{$v_i$} ;
		\draw[dashed] (0,0) -- (-1,2) -- (-1,10) -- (0,8);
	\end{tikzpicture}
 \\ \hline
(E4) A nodal line curves downwards & 
	\begin{tikzpicture}[scale=0.25]
		\draw[ultra thick] (1,1) .. controls (2.5,4) .. (4,1);
		\draw[]  (5,-1) -- (0,0) -- (0,8) -- (5,7)  ;
		\draw[dotted, very thick]  (5,-1) -- (5.5, -0.5) -- (4.5, 0.5) -- (5.5, 1.5) -- (4.5,2.5) -- (5.5,3.5) -- (4.5,4.5) -- (5.5,5.5) -- (4.5,6.5) -- (5,7);
		\draw[dashed]  (0,0) -- (-3,-1) -- (-3,7) -- (0,8) node[anchor=north west]{$v_i$} ;
		\draw[dashed] (0,0) -- (-1,2) -- (-1,10) -- (0,8);
	\end{tikzpicture}\\ \hline
(E5) A nodal line hits a boundary vertex & 
	\begin{tikzpicture}[scale=0.25]
		\draw[ultra thick] (2,1)-- (5,5);
		\draw[ultra thick] (5,-1) -- (5,7);
		\draw[]  (0,0) -- (5,-1) --  (5,7) node[anchor=north west]{$v_b$} -- (0,8) -- (0,0);
		\draw[dashed]  (0,0) -- (-3,-1) -- (-3,7) -- (0,8)  ;
		\draw[dashed] (0,0) -- (-1,2) -- (-1,10) -- (0,8);
	\end{tikzpicture} \\ \hline
(E6) A nodal line hits an inner vertex & 

	\begin{tikzpicture}[scale=0.25]
		\draw[ultra thick] (3,1)-- (0,5);
		\draw[]  (5,-1) -- (0,0) -- (0,8) -- (5,7)  ;
		\draw[dotted, very thick]  (5,-1) -- (5.5, -0.5) -- (4.5, 0.5) -- (5.5, 1.5) -- (4.5,2.5) -- (5.5,3.5) -- (4.5,4.5) -- (5.5,5.5) -- (4.5,6.5) -- (5,7);
		\draw[dashed]  (0,0) -- (-3,-1) -- (-3,7) -- (0,8) node[anchor=north west]{$v_i$} ;
		\draw[dashed] (0,0) -- (-1,2) -- (-1,10) -- (0,8);
		\draw[dashed, ultra thick] (0,5) -- (-0.5,7.5);
	\draw[dashed, ultra thick] (0,5) -- (-2.5, 6);
	\end{tikzpicture}
	\, \, 
	\begin{tikzpicture}[scale=0.25]
		\draw[ultra thick] (3,6)-- (0,4);
		\draw[]  (5,-1) -- (0,0) -- (0,8) -- (5,7)  ;
		\draw[dotted, very thick]  (5,-1) -- (5.5, -0.5) -- (4.5, 0.5) -- (5.5, 1.5) -- (4.5,2.5) -- (5.5,3.5) -- (4.5,4.5) -- (5.5,5.5) -- (4.5,6.5) -- (5,7);
		\draw[dashed]  (0,0) -- (-3,-1) -- (-3,7) -- (0,8) node[anchor=north west]{$v_i$} ;
		\draw[dashed] (0,0) -- (-1,2) -- (-1,10) -- (0,8);
		\draw[dashed, ultra thick] (0,4) -- (-0.5,1.5);
	\draw[dashed, ultra thick] (0,4) -- (-2.5, 2);
	\end{tikzpicture}\\
\hline
\end{tabular}}
\end{center}

\newpage

We will look at the effects of events E1 to E5 on $N(F_y)$ as $y$ increases.

\begin{lemma}\label{E1E2}
	Events E1 and E2 cannot happen inside an edge. Also, two nodal lines cannot emanate from a vertex inside the same edge, unless the edge is a loop. In that case, at most one nodal line can emanate in each lead from the vertex.
\end{lemma}

\begin{proof}
This is an immediate corollary of Proposition \ref{maximumprinciple}. 
\end{proof}

\begin{lemma}\label{E3E4}
	Events E3 and E4 can only decrease the nodal count of $F_y$ as $y$ increases.
\end{lemma}

\begin{proof}
Assume that $k$ nodal lines collide at a single point $(x_{0},y_{0})$ with $k \geq 2$. Proposition \ref{maximumprinciple} guarantees that either one or zero nodal line will emanate from $(x_0,y_0)$. Therefore, this can only decrease the nodal count.

\end{proof}

\begin{lemma}\label{E5}
	The event E5 can only decrease the nodal count.
\end{lemma}

\begin{proof}
Assume that a nodal line hits a boundary vertex. By Proposition \ref{maximumprinciple}, the only nodal line that can emanate from a boundary vertex $v_b$ as $y$ increases is one that is due to the boundary condition. Therefore, the event E5 decreases the nodal count.

\end{proof}

\subsection{Nodal line hitting an inner vertex (event E6)}\label{section:innervertex}\hfill
\medskip

We now examine what happens when one (or more) nodal line hits an inner vertex $v$. Let us fix an inner vertex $v$ and set $g_v(y) := g(v,y)$. This is well-defined since $g$ is continuous at inner vertices. There exist coefficients $a_i(v)$ such that $g_v(y)=\sum_{i=1}^{M}a_{i}(v)e^{-\la_{i}y}$.

\medskip
A nodal line hits $v$ at $(v,y)$ when $g_v(y) =0$

\medskip

First, we will count the number of times that these events can happen, i.e. the number of zeroes of $g_v$.

\medskip

We will use the following bound on the number of zeroes of a linear combination of
exponentials:

\begin{thm}\label{thmszego}
\cite[part V, Chapter 1, problem 77]{PolyaSzego}
	Let $h(x)=\sum_{i=1}^{n} a_{i}e^{b_{i}x}$ with $a_{i}\neq0$ and $b_{i+1}>b_{i}$ for any $i \leq n$.
	Let $C(h)$ be the number of times that $a_{i+1}$ and $a_{i}$ have
	different signs. Then, the number of zeroes of $h$ is less or equal
	than $C(h)$.
\end{thm}

%
%

\medskip

This implies that for any $v \in V_i$, $g_v$ has at most $M-1$ zeroes. We now look at what can happen to $N(F_y)$ when $y$ increases if $g_v(y)=0$ for some inner vertex $v$.

\begin{lemma}\label{lemma:inner}
	When a nodal line hits an inner vertex $v$, the nodal count can increase as $y$ increases by at most $deg(v)-2$.
\end{lemma}

\medskip

\begin{proof}

For every edge $e \in \mathcal{E}_v$, we will reparametrize $e$ as $[0,l_e]$ such that $v$ is identified with $0$.

\medskip

As zeroes of $g_v$ are discrete by theorem \ref{thmszego}, we can assume that $g_v(y')=0$ and that for some $\ep>0$, $g_v(y) \neq 0$ if $y \in (y'-\ep, y') \cup (y', y'+\ep)$.   We will also choose $\ep$ smaller than the smallest edge of $\Gamma$ in order for the map $ \theta \to \left( \ep \cos(\theta), y' + \ep \sin(\theta)\right) \in e \times \mathbb{R}$ to be well-defined for any $e \in \mathcal{E}_v$.

\medskip

Now, assume that $g_v$ changes sign at $y=y'$. By the Neumann-Kirchoff continuity assumptions at inner vertices, it means that on each edge $e \in \mathcal{E}_v$, \linebreak $g|_e(\ep \cos(\theta), y' + \ep \sin(\theta))$ has changed sign $2m+1$ times from $\theta=-\pi/2$ to $\theta = \pi/2$ for some $m\geq 0$. By lemma \ref{E1E2}, we know that for any $e \in \mathcal{E}_v$, two nodal lines cannot emanate from $(0,y)$. This implies that for $e \in \mathcal{E}_v$ either $2m$ nodal lines hit $(0,y)$ and one comes out or $2m+1$ edges hit $(0,y)$ and none come out. Therefore, the maximum increase of the nodal count happens if one nodal line hits $v$ from a single edge in $\mathcal{E}_v$ and it comes out on every other edge in $\mathcal{E}_v$. This increases the nodal count by $deg(v)-2$, which is illustrated in the next picture:

\begin{figure}[h]
	\begin{tikzpicture}[scale=0.3]
	\draw[ultra thick] (3,1)-- (0,5);
	\draw[]  (0,0) -- (5,-1) -- (5,7) -- (0,8) -- (0,0);
	\draw[dashed]  (0,0) -- (-3,-1) -- (-3,7) -- (0,8) node[anchor=north west]{$v_i$} ;
	\draw[dashed] (0,0) -- (-1,2) -- (-1,10) -- (0,8);
	\draw[dashed] (0,0) -- (-3,0.5) -- (-3,8.5) -- (0,8);
	\draw[dashed, ultra thick] (0,5) -- (-0.5,7.5);
	\draw[dashed, ultra thick] (0,5) -- (-2.5, 6);
	\draw[dashed, ultra thick] (0,5) -- (-2, 7);
\end{tikzpicture}
\caption{Event of maximum nodal increase}
\end{figure}
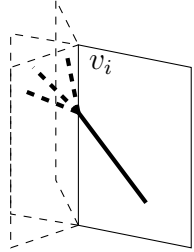

\medskip

Now, assume that $g(v,y)$ does not change sign at $y=y'$. Again, by continuity at inner vertices, for each $e \in \mathcal{E}_v$ we know that $g|_e(\ep \cos(\theta), \ep \sin(\theta))$ changes sign $2m_e$ times from $\theta=-\pi/2$ to $\theta = \pi/2$ for some $m_e \geq 0$ on each edge $e \in \mathcal{E}_v$. Since at most one nodal line can come out the vertex on any edge in $\mathcal{E}_v$, this event cannot increase the nodal count. This completes the proof of Lemma \ref{lemma:inner}.

\end{proof}

\medskip

Combining Theorem \ref{thmszego} and Lemma \ref{lemma:inner}, we get the following characterization of nodal lines hitting inner vertices:

\begin{lemma}\label{E6}
	Event E6 can happen at most $M-1$ times at each inner vertex, and each time it happens can increase the nodal count by at most $\deg(v)-2$.
\end{lemma}

\medskip

\subsection{Global bounds}\label{sectfinal}\hfill

We will combine the lemmas of sections \ref{section:events} and \ref{section:innervertex}.

We start with $N(f_{M})$ zeroes at $y= -\infty$. As $y$ increases, by Lemmas \ref{E1E2}, \ref{E3E4}, \ref{E5} and \ref{E6} the only event which increases the number of nodal lines is a crossing at an inner vertex. For each inner vertex $v$, there can be at most $M-1$ such crossings,
and each crossing can create at most $deg(v)-2$ new nodal lines by Lemma \ref{E6}.

\medskip

If we now start at $y=+\infty$, we have $N(f_{k_{1}})$ zeroes. When
we decrease $y$, nodal lines cannot merge together or hit a boundary
vertex by lemmas \ref{E1E2} and \ref{E5}. Also, by Lemma \ref{E3E4}, nodal lines can be created or split but it only increases the nodal count. Therefore, the only event that may decrease the number of
zeroes is nodal lines colliding at an inner vertex. By Lemma \ref{E6}, the biggest decrease for a degree $d$ vertex is when $d-1$ nodal lines collide
and only one exits. This can happen at most $M-1$ times at each vertex.

\medskip

This gives us the following two-sided bound:

\begin{equation}\label{lemma:global1}
N(f_{k_{1}})-(M-1)\sum\limits_{v \in V_i}(\deg(v)-2) \leq N(F_0) \leq N(f_{k_{M}})+(M-1)\sum\limits_{v \in V_i}(\deg(v)-2) \, .
\end{equation}

We are ready to finish the proof of Theorem \ref{thm1}.

\begin{proof}[Proof of Theorem \ref{thm1}]

We know from \cite[Theorem 5.28]{BerKuc_graphs} that $k-1 \leq N(f_k) \leq k-1+\beta$ for any eigenfunction $f_k$ of $H_W$. Combining theses inequalities with equation \ref{lemma:global1} gives us the following:

\begin{equation}\label{lemma:global}
k_1-1-(M-1)\sum_{v \in V_i}(deg(v)-2) \leq N(F_0) \leq k_M-1+\beta+(M-1)\sum_{v \in V_i}(deg(v)-2)\, .
\end{equation}

Now, recalling that $\beta = |E|-|V|+1$, we have that the following:

\begin{align}\label{topo}
\sum\limits_{v \in V_i}(deg(v)-2) &= \sum\limits_{v \in V}deg(v) - |V_b| - 2 |V_i| \, ,\nonumber \\
&= |V_b| + \left( 2 |E| - 2|V_b| - 2|V_i| \right) \, , \nonumber \\
&= |V_b| + 2\beta -2 \, .
\end{align}

Combining \eqref{lemma:global} and \eqref{topo} completes the proof of Theorem \ref{thm1}.

\end{proof}

\begin{remark}
It is clear from the proof that the condition that $\Gamma$ is $W$-generic can be replaced by the condition that the eigenfunctions in the linear combination are non-zero at inner vertices.
\end{remark}

\section{Saturating examples for the upper bound - proof of Theorem \ref{thm3}}\label{section:upper}

Let $G(s,\ep)$ be a star graph with one edge of length $1$ and $s$
edges of length $\ep$.

\begin{figure}[h]
	\centering
	\begin{tikzpicture}[scale=1]
		\draw[] (-5,0)-- (0,0);
			\draw[]  (0,0)--(-0.5,0.5);
				\draw[]  (0,0)--(-0.5,-0.5);
		\draw[]  (0,0)--(0,0.7);
		\draw[]  (0,0)--(0.5,0.5);
		\draw[]  (0,0)--(0.7,0);
		\draw[]  (0,0)--(0.5,-0.5);
		\draw[]  (0,0)--(0,-0.7);
	\end{tikzpicture}
	\caption{$G(s,\ep)$ for $s=7$}
\end{figure}

\medskip

This graph has some interesting properties if $\ep$ is taken small enough:

\begin{lemma}\label{lemmaGme}
	For any $s$ and $M$ we can take $\ep$
	small enough such that the following occurs:
	\begin{itemize}
	\item[(a)] The first $M$ eigenvalues of $G(s,\ep)$ are all simple.
	\item[(b)] The first $M$ eigenfunctions are all invariant with	respect to permutations of the small edges.
	\item[(c)] For any $n \leq M$, each  eigenfunction $f_{n}$ of $-\frac{\partial^2}{\partial x^2}$ has exactly $n-1$ zeroes on the long edge and no zeroes on the small
	edges or on the inner vertex.
\end{itemize}	
\end{lemma}

\medskip

\begin{proof}
By \cite[Theorem $4.5$]{BERKOLAIKO2019632}, as $\ep \to 0$ the eigenvalues of $G(s,\ep)$ converge to those of the unit interval. Therefore, for any integer $M$ and $\alpha>0$ small, there exists $\ep_0 >0$ small enough such that for any $0<\ep<\ep_0$ and $1 \leq n \leq M+1$, $|\lambda_n(G(s,\ep)) - \pi^2n^2 | < \alpha$. 

\medskip

As a consequence, the first $M$ eigenvalues are simple for any $\ep <\ep_0$. This completes the proof of the first part of Lemma \ref{lemmaGme}.

\medskip

Let $v$ be the inner vertex of $G(s,\ep)$. Fix $\ep < \ep_0$ such that for $1 \leq n \leq M$, $\ep \sqrt{\lambda_n} < \pi/2$. Since the restriction of $f_n$ to any small edge is equal to $C\sin( \sqrt{\lambda_n})x)$ up to a constant, this ensures that $f_n(v) \neq 0$ for any $n \leq M$. Furthermore, since $f_n(v)$ is well-defined by the continuity assumptions on the graph, this constant is the same for any small edge. This implies that for any $n \leq M$, $f_n$ is invariant with respect to permutations of the small edges. This completes the proof of the second part of Lemma \ref{lemmaGme}.

\medskip

By \cite{Ban_ptrsa14}, $f_n$ has exactly $n-1$ zeroes. However, since $\ep \sqrt{\lambda_n} < \pi/2$, $f_n$ does not have a zero inside any small edge, which imples that $f_n$ has $n-1$ zeroes on the long edge. This completes the proof of the third part of lemma \ref{lemmaGme}.

\end{proof}

\medskip

We will now construct linear combinations of eigenfunctions of $G(s,\ep)$ with a high nodal count:

\begin{prop}
	\label{prop1}~
	
	For any $M,s>0$, there exists $\ep_1(M,s)$ small enough such that
	for any $\ep<\ep_1(M,s)$ and any $L\leq M$, there exist linear combinations
	of the first $L$ eigenfunctions of $G(s,\ep)$ with exactly $L-1+(L-1)(s-1)$ zeroes on the small edges.
\end{prop}

\medskip
\begin{proof}

We choose $\ep_1(M,s)$ such that Lemma \ref{lemmaGme} applies and $\ep < \ep_1(M,s)$. For any  $L\leq M$ we can choose $F:=\sum\limits_{n=1}^{L}a_{n}f_{n}$ such that $F$ has $L-1$ zeroes on
a chosen small edge of $G(s,\ep)$. This is a consequence of the linear independance of the first $M$ eigenfunctions on any small edge, since by Lemma \ref{lemmaGme} they are all non-zero on the small edges. Since these $f_{n}$ are symmetric with respect to permutations of the small edges, so is $F$. Therefore, $F$ has exactly $(L-1)s=L-1+(L-1)(s-1)$ zeroes on the small edges. \end{proof}

\medskip
Now, we note that this graph is not generic, since it is possible to find eigenfunctions which are zero at the inner vertex (for instance by choosing $\lambda = \pi^2 \ep^{-2}$). In order to find a saturating example in the set of $0$-generic graphs, we will use the fact that for a given graph, there exists an arbitrarily small perturbation of the edge lengths such that the spectrum of the Laplacian is simple and no eigenfunction vanishes on a vertex \cite[Theorem $3.6$]{BerLiu_jmaa17}.

\medskip
We will construct a perturbation of $G(s,\ep)$ and linear combinations of eigenfunctions that have the same behaviour as in proposition \ref{prop1}:

\begin{lemma}\label{lemma5}
There exists a graph $G_\delta(s,\ep)$ which has the following properties:
\begin{itemize}
\item $G_\delta(s,\ep)$ is obtained by perturbing the edge length of $G(s,\ep)$ by at most $\delta$.
\item $G_\delta(s,\ep)$ is $0$-generic.
\item For any $L \leq M$, there exists a linear combination $F_\delta$ of the first $L$ eigenfunctions on $G_\delta(s,\ep)$ such that $N(F_\delta) = L-1 +(L-1)(s-1)$.
\end{itemize}
\end{lemma}

This construction immediately implies Theorem \ref{thm3}.
\begin{proof}

First, we choose $\ep>0$ such that Lemma \ref{lemmaGme} applies to $G(s,\ep)$.

	Now, for any $\delta >0$, it is possible to find a graph $G_\delta(s,\ep)$ with one edge of length $1$ and $s$ edges $e_i$ of length $l_{e_i} \in (\ep, \ep+\delta)$ such that $G_\delta(s,\ep)$ is $0$-generic. Furthermore, as $\delta \to 0$, the eigenvalues of $G_\delta(s,\ep)$ converge to those of $G(s,\ep)$ (see for instance \cite[Appendix A]{BanLev_prep16}, \cite[Theorem $4.15$]{rohleder2023spectral} or \cite[Theorem $3.6$]{BERKOLAIKO2019632}). 
	
	\medskip
	
	  We define the map $\phi_\delta : G_\delta(s,\ep) \to G(s,\ep)$ that fixes the long edge and sends $x \in e_i$ to $(\ep/l_{e_i}) x$.

\medskip	
	
	Let $f_{\delta,n}$ be the $n$-th eigenfunction on $G_\delta(s,\ep)$. 
	
	\medskip
	We know from \cite[Theorem 3.1.4]{BerKuc_graphs} that eigenfunctions depend analytically on perturbations of edge lengths. Therefore, since the first $M$ eigenvalues of $G(s,\ep)$ are simple, as $\delta$ goes to zero, $\sup\limits_{x \in G_\delta(s,\ep)}|f_n(\phi_\delta(x))) - f_{\delta,n}(x)|$ will go to zero for any $1 \leq n \leq M$.
	
	\medskip

Now, for $1 \leq L \leq M$, take a linear combination $F= \sum\limits_{n=1}^L a_n f_n$ with $L-1$ zeroes on each small edge of $G(s,\ep)$. We define  $F_\delta : G_\delta(s,\ep) \to \mathbb{R}$, $ F_\delta := \sum\limits_{n=1}^N a_n f_{\delta,n}$ with the same coefficients $a_n$. 

\medskip

As $\delta \to 0$, $\sup\limits_{x \in G_\delta(s,\ep)}|F_\delta(x) - F(\phi_\delta(x))| \to 0$.

\medskip

 Therefore, if we take $\delta$ small enough, $F_\delta$ will have at least as many zeroes as $F$ on each small edge by the mean value theorem. Also , $F_\delta$ cannot have any more zeroes since it saturates the upper bound in theorem \ref{thm1} and $G_\delta(s,\ep)$ is $0$-generic. Hence, $N(F_\delta) = (L-1)s$, which completes the proof of lemma \ref{lemma5} and of theorem \ref{thm3}.
	
\end{proof}

\section{Examples of non-trivial lower bounds - proof of Theorem \ref{thm6}}

Now, let $I(m,\ep)$ be the following graph: start with two edges $e_1$ and $e_2$ of length $1/2$ and connect them with $m$ parallel edges of length $\ep$. This graph has $2$ boundary vertices $v_1$ and $v_4$ and $2$ inner vertices $v_2$ and $v_3$ of degree $m+1$.
\medskip

We will define the involution $\psi : I(m,\ep) \to I(m,\ep)$ as the reflection across the dotted line in figure \ref{Imep}. 

\pagebreak

\begin{figure}[h]
	\begin{tikzpicture}[scale=0.7]
		\draw[] (-4.5,0.6) node{$v_1$};
		\draw[] (-1.5,0.6) node{$v_2$};
		\draw[] (1.5,0.6) node{$v_3$};
		\draw[] (4.5,0.6) node{$v_4$};
		\draw[] (-2.5,1.1) node{$e_1$};
		\draw[] (2.5,1.1) node{$e_2$};
		\draw[] (-4,0)-- (-1,0);
		\draw[] (4,0)-- (1,0);
		\draw[] (-1,0) .. controls(0,1) .. (1,0) ;
		\draw[] (-1,0) .. controls(0,0.5) .. (1,0) ;
		\draw[] (-1,0) .. controls(0,0) .. (1,0) ;
		\draw[] (-1,0) .. controls(0,-0.5) .. (1,0) ;
		\draw[] (-1,0) .. controls(0,-1) .. (1,0) ;
		\draw[thick, dotted] (0,-3) -- (0,3);
	\end{tikzpicture}
\caption{$I(m,\ep)$ for $m=5$}\label{Imep}
\end{figure}
	
\medskip

\begin{lemma}\label{prop2}
	For any $m$, there exists $\ep$ small enough such that $I(m,\ep)$ has the following properties:
	\begin{itemize}
	\item[(a)] The first three eigenvalues are simple.
	\item[(b)] $f_2 \circ \psi = - f_2$ and $N(f_2)=m$.
\item[(c)] $f_3 \circ \psi = f_3$  and $N(f_3)=2$.
\end{itemize}
\end{lemma}

\begin{proof}

As in the proof of Lemma \ref{lemmaGme}, by \cite{BERKOLAIKO2019632}, \cite{BanLev_prep16} and \cite{rohleder2023spectral}, when $\ep$ goes to zero the eigenvalues of $I(m,\ep)$ converge to the eigenvalues of the unit interval with Dirichlet boundary conditions. Since all the eigenvalues on the interval are simple, by taking $\ep$ small enough the first three eigenvalues of $I(m,\ep)$ are simple. This completes the proof of the first part of Lemma \ref{prop2}.

Since the first three eigenvalues are simple, for $i=1,2,3$, $f_i \circ \psi = \pm f_i$.

\medskip

Let us define the graph $\frac{1}{2} G(m,\ep)$, which is constructed by taking the graph $G(m,\ep)$ that was defined in section \ref{section:upper} and dividing the length of every edge by two.

\medskip

We notice that $I(m,\ep)$ is made by gluing two copies of the rescaled graph $\frac{1}{2} G(m,\ep)$  along the small edges (and the boundary conditions removed at the gluing points since these points become vertices of degree two). Therefore, if an eigenfunction on $I(m,\ep)$ is zero at the center of each small edge, then its restriction to $\frac{1}{2} G(m,\ep)$ is an eigenfunction on $\frac{1}{2} G(m,\ep)$.

\medskip

The first eigenfunction on $\frac{1}{2} G(m,\ep)$, which we extend to an odd function $\tilde{f}$ with respect to $\psi$, is an eigenfunction on $I(m,\ep)$. The first eigenvalue on $\frac{1}{2} G(m,\ep)$ converges to $4\pi^2$, and the first and third eigenvalues on $I(m,\ep)$ tend to $\pi^2$ and $9 \pi^2$ respectively as $\ep \to 0$. It means that for $\ep$ small enough, $\tilde{f}$ is the second eigenfunction $f_2$ on $I(m,\ep)$. Therefore, $f_2 \circ \psi = -f_2$ and its only zeroes are at the center of each small edge. This completes the proof of the second part of Lemma \ref{prop2}.

\medskip

Finally, if $f_3 \circ \psi = - f_3$, then $f_3$ would have to have a zero in the middle of each small edge. Therefore, it would be an eigenvalue of $\frac{1}{2} G(m,\ep)$. However, as $\ep \to 0$, the first two eigenvalues of $\frac{1}{2} G(m,\ep)$ tend to $4\pi^2$ and $16 \pi^2$, while the third eigenvalue of $I(m,\ep)$ tends to $9 \pi^2$. This implies that $f_3 \circ \psi = f_3$. Also, if $f_3$ had a zero on a small edge, by symmetry it would have at least two zeroes on each small edge, which is only possible if $\ep \sqrt{\lambda_3} \geq \pi/2 $. By taking $\ep$ small enough, this cannot happen. Therefore, $f_3$ has one zero on each long edge and no other zeroes. This completes the proof of the third part of Lemma \ref{prop2}.

\end{proof}

\medskip

We now construct a linear combination of $f_2$ and $f_3$ with only one zero:

\begin{lemma}\label{lemmalower}
There exist $\ep$ small enough and $b \in \mathbb{R}$ such that $N(f_2 + bf_3)=1$ on $I(m,\ep)$.
\end{lemma}

\begin{proof}
We now fix the normalization of $f_2$ and $f_3$ such that $\partial_{e_1}f_2(v_1) = \partial_{e_1}f_3(v_1)=1$. This makes $f_3$ strictly negative on the small edges.

\medskip

There exists $C>0$ such that $\sin(2 \pi x) -C \sin(3 \pi x) > 0 $ for $x \in (0,1/2)$ and \linebreak $\sin(2 \pi x) + C \sin(3 \pi x)$ has exactly one zero inside $(0,1/2)$. We recall that the eigenfunctions of $I(m,\ep)$ converge pointwise to that of the interval as $\ep \to 0$, that $f_2$ is antisymmetric and $f_3$ is symmetric. This implies that for $\ep_0$ small enough there exists $C_1>0$ such that $f_2 - C_1 f_3$ has no zero on $e_1$ and one zero on $e_2$ for any $\ep < \ep_0$. 

\medskip
We will now show that $f_2 - C_1 f_3$ has no zeroes on the small edges. As $\ep \to 0$, by eigenfunction convergence the supremum of $|f_2|$ on the small edges converges to zero, while the infimum of $|f_3|$ on the small edges converges to some strictly positive value $C_2$ (the constant $C_2$ can be computed explicitely but we will not do it here). Now, choose $\ep < \ep_0$ such that the infimum of $|f_3|$ on the small edges is greater than $C_2 /2$ and the supremum of $f_2$ on the small edges is smaller than $C_1 C_2 /4$. With this choice of $\ep$, for any $x$ in a small edge, $$f_2(x) - C_1 f_3(x) > -C_1 C_2 /4 + C_1 C_2 /2 > 0 \, .$$ 

\medskip

Therefore, $f_2 - C_1 f_3$ has no zero on $e_1$ or any small edge and exactly one zero on $e_2$. This completes the proof of Lemma \ref{lemmalower}.
\end{proof}

Since the graph $I(m,\ep)$ is not $0$-generic, we will slightly perturb it without changing the nodal count to complete the proof of Theorem \ref{thm6}.

\medskip
\begin{proof}[Proof of Theorem \ref{thm6}]

As in the proof of Lemma \ref{lemma5}, we can construct a $\delta$-small perturbation ${I}_\delta(m,\ep)$ of $I(m,\ep)$ which is $0$-generic. Let $f_{\delta,n}$ be the $n$-th eigenfunction on $I_\delta(m,\ep)$. By a similar argument to the one in the proof of Lemma \ref{lemma5}, we can choose $\delta$ small enough such that $N(f_{\delta,2})=N(f_2) = m$, $N(f_{\delta,3})=N(f_3)=2$ and $N(f_{\delta,2} -C_1 f_{\delta,3}) = N(f_2-C_1f_3) = 1$, which proves theorem \ref{thm6}.
\end{proof}


\bibliography{GlobalBib_210709,newbibphil}
\bibliographystyle{plain}

\end{document}